\documentclass[pdflatex,sn-mathphys]{sn-jnl}

\jyear{2021}%

\theoremstyle{thmstyleone}%
\newtheorem{theorem}{Theorem}
\newtheorem{lemma}{Lemma}
\newtheorem{proposition}[theorem]{Proposition}%

\theoremstyle{thmstyletwo}%
\newtheorem{example}{Example}%
\newtheorem{remark}{Remark}%

\theoremstyle{thmstylethree}%
\usepackage{amsthm}
\usepackage{amsmath}
\usepackage{amssymb}
\usepackage{url}
\usepackage{float}
\usepackage{booktabs}
\usepackage{bm}

\begin{document}

\title[Article Title]{New Binary Quantum Codes Constructed from Quasi-Cyclic Codes}


\author*[1]{\fnm{Chaofeng} \sur{Guan}}\email{gcf2020yeah.net}

\author[1]{\fnm{Ruihu} \sur{Li}}\email{liruihu@aliyun.com}
\equalcont{These authors contributed equally to this work.}

\author[1]{\fnm{Liangdong } \sur{Lu}}\email{kelinglv@163.com}
\equalcont{These authors contributed equally to this work.}

\author[1]{\fnm{Yu } \sur{Yao}}\email{njjbxqyy@163.com}
\equalcont{These authors contributed equally to this work.}

\affil*[1]{\orgdiv{Fundamentals Department}, \orgname{Air Force Engineering University}, \orgaddress{ \city{Xi'an}, \state{Shaanxi}, \postcode{710051}, \country{ P. R. China}}}


\abstract{It is well known that quantum codes can be constructed by means of classical symplectic dual-containing codes.
This paper considers a family of two-generator quasi-cyclic codes and derives sufficient conditions for these codes to be symplectic dual-containing.
Then, a new method for constructing binary quantum codes using symplectic dual-containing codes is proposed. 
As an application, we construct 8 binary quantum codes that exceed the best-known results.
Further, another 36 new binary quantum codes are obtained by propagation rules, all of which improve the lower bound on the minimum distances.}



\maketitle
\section{Introduction}
Quantum error-correcting codes (QECCs) play an essential role in safeguarding quantum information from being corrupted by undesirable environmental and operational noise (decoherence).
Since the first binary $[[9,1,3]]$ QECC was derived by Shor \cite{shors1995scheme} in 1995, many scholars have extensively investigated the construction of QECCs that perform well on quantum error-correction.
In \cite{steane1996multiple,steane1999enlargement,calderbank1998quantum,ketkar2006nonbinary}, the association between QECCs and classical codes with certain self-orthogonal or dual-containing parameters were established, QECCs can be obtained from Euclidean, Hermitian, symplectic self-orthogonal or dual-containing codes.
Meanwhile, a large number of QECCs with good parameters over finite fields were constructed using cyclic codes, constacyclic codes, quasi-cyclic codes, etc. \cite{2009QUANTUM,Kai2014ConstacyclicCA,2017Improved,Gao2018QuantumCD,Song2018TwoFO,Li2019NewQC,Lv1,Lv2,Lv3,Yao1,Yao2}.

As an effective extension of cyclic codes, quasi-cyclic (QC) codes are an important class of linear codes with good algebraic structure.
In \cite{kasami1974gilbert}, Kasami et al. showed that QC codes satisfy the modified Gilbert-Vashamov (GV) bound, so there are many good linear codes in this structure.
Since the excellent performance of QC codes, many scholars have applied them to construct QECCs with suitable parameters.
In 2007, Hagiwara et al. \cite{hagiwara2007quantum} first used QC codes to construct QECCs.
In 2018, Galindo et al. \cite{galindo2018quasi} initially proposed a method for constructing QECCs by dual-containing QC codes under different inner products.
Thereby, Ezerman et al.\cite{ezerman2019good} employed quantum Construction X to QC codes with large Hermitian hulls and derived a new binary $[[31,9,7]]$ QECC.
In addition, by studying several classes of QC codes, Lv et al. \cite{Lv1,Lv2,Lv3} obtained many binary QECCs with better performance than the best-known ones through symplectic and Hermitian constructions.
In \cite{Yao1} and \cite{Yao2}, the work of Lv et al. were further advanced by Yao et al. to determine the self-orthogonal conditions for two classes of  quasi-twisted codes and constructed two record-breaking binary QECCs.
Motivated by the work mentioned above, in this paper, we will put forward a new structure of two-generator quasi-cyclic (2-QC) codes by which many record-breaking binary QECCs are derived.

This paper is structured as follows.
In Sect. \ref{sec2}, some preliminary concepts are given. In Sect. \ref{sec3}, we present a particular structure of 2-QC codes. By means of studying algebraic structure of their dual codes under symplectic inner product, sufficient conditions for these 2-QC codes to be symplectic dual-containing are determined.
In Sect. \ref{sec4}, 44 new binary QECCs are obtained by utilizing the symplectic construction and propagation rules.
Finally, Sect. \ref{sec5} concludes this paper.
\section{Preliminaries} \label{sec2}
A non-empty subset  $\mathscr{C}$  of  $F_{2}^{2n}$  is called a code of length  $2n$  over  $F_{2}$. 
If dimension of $\mathscr{C}$ is $k$ and its minimum distance is $d$, then $\mathscr{C}$ can be written as an $[2n,k,d]$ code. 
For $\vec{u}=(u_{0},  \ldots, u_{2n-1})$, 
$\vec{v}=(v_{0},\ldots, v_{2n-1}) \in F_2^{2n}$,  Euclidean and symplectic inner product of them can be defined as $\langle\vec{u}, \vec{v}\rangle_{e}=\sum_{i=0}^{2 n-1} u_{i} v_{i}$ and $\langle\vec{u}, \vec{v}\rangle_{s}=\sum_{i=0}^{n-1}\left(u_{i} v_{n+i}-u_{n+i} v_{i}\right)$, respectively. 
The Euclidean and symplectic dual codes of $\mathscr{C}$ can be separately denoted as $\mathscr{C}^{\perp_{e}}=\left\{\vec{v} \in F_{q}^{2 n} \mid\langle\vec{u}, \vec{v}\rangle_{e}=0, \forall \vec{u} \in \mathscr{C}\right\}$ and $ \mathscr{C}^{\perp_{s}}= 
\left\{\vec{v} \in F_{q}^{2n} \mid\langle\vec{u}, \vec{v}\rangle_{s}=0, \forall \vec{u} \in \mathscr{C}\right\}.$
If $\mathscr{C}^{\perp_{e}} \subset \mathscr{C}$, then we can say $\mathscr{C}$ is a Euclidean dual-containing code. If $\mathscr{C}^{\perp_{s}} \subset \mathscr{C}$, then $\mathscr{C}$ is a symplectic dual-containing code.
The symplectic weight of vector $\vec{u} \in F_{2}^{2n}$ is 
$w_{s}(\vec{u})=\operatorname{card}\left\{i \mid (u_{i}, u_{n+i}) \neq(0,0)\right\} \text { for } 0 \leq i \leq n-1$.
The minimum symplectic distance of $\mathscr{C}$ can be written as $d_{s}(\mathscr{C})=\min \left\{w_{s}(\vec{u}) \mid \vec{u} \in \mathscr{C}\right\}.$

Let $\mathscr{C}_1$ be a code of length $n$ over $F_2$. $\mathscr{C}_1$ is a cyclic code if it is closed under a cyclic shift, i.e.,  for any $c=\left(c_{0}, c_{1}, \ldots, c_{n-1}\right)\in \mathscr{C}_1$, then 
$c^\prime =\left(c_{n-1}, c_{0}, \ldots, c_{n-2}\right)\in \mathscr{C}_1$  as well. 
In addition, $\mathscr{C}_1$ is isomorphic to an ideal of the principal quotient ring  $\mathcal{R}=F_{2}[x] /\left\langle x^{n}-1\right\rangle$.
If $\mathscr{C}$ is generated by a monic divisor  $g(x)$  of  $x^{n}-1$, i.e., $\mathscr{C}_1=\langle g(x)\rangle$ and $g(x)\mid x^n-1$. $g(x)$ is called generator polynomial of $\mathscr{C}_1$.
Let  $ \mathbb{Z} _{n}=\{0,1, \ldots, n-1\}$  and  $\zeta $  be a primitive  $n$-th root of unity in some extended fields of  $F_{2}$.  
Then defining set of $\mathscr{C}_1=\langle g(x)\rangle$  can be written as 
$T=\left\{i \in  \mathbb{Z} _{n} \mid g\left(\zeta ^{i}\right)=0\right\}$. 
If $C_{i}$ is a $2$-cyclotomic coset modulo  $n$, then it can be denoted as
$C_{i}=\left\{i, i 2, \ldots, i 2^{s-1} \mid i \in \mathbb{Z} _{n}\right\}$, where  $s$  is the smallest positive integer  with  $i 2^{s} \equiv i \bmod n$.

If $\mathscr{C}$ is a QC code of length $2 n$ with index 2, 
then for any codeword $c=(c_{0}, \ldots, c_{n-1}, $ $c_{n}, \ldots, c_{2 n-1})$ of $\mathscr{C}$, there exists $c^\prime =( c_{n-1},  c_{0}, \ldots,$ $ c_{n-2},c_{2 n-1}, c_{n}, \ldots, c_{2 n-2})$ $\in$ $\mathscr{C}$.
Similarly, $\mathscr{C}$  is isomorphic to an $\mathcal{R}$-submodule of  $\mathcal{R}^2$. 
Circulant matrices are basic components in generator matrix for QC codes. An $n\times n$ circulant matrix is defined as

$$
A=\left(\begin{array}{ccccc}
a_{0} & a_{1} & a_{2} & \ldots & a_{n-1} \\
a_{n-1} & a_{0} & a_{1} & \ldots & a_{n-2} \\
\vdots & \vdots & \vdots & \vdots & \vdots \\
a_{1} & a_{2} & a_{3} & \ldots & a_{0}
\end{array}\right).$$

If the first row of  $A$  is mapped onto polynomial  $a(x)$, then circulant matrix  $A$  is isomorphic to polynomial $a(x)=a_{0}+a_{1} x+\cdots+a_{n-1} x^{n-1} \in \mathcal{R}$. 
Generator matrix of a 2-QC code with index 2 can be transformed into rows of $n \times n$ circulant matrices by suitable permutation of columns, which has the following form:
$$
G=\left(\begin{array}{cccc}
A_{1,1} & A_{1,2}\\
A_{2,1} & A_{2,2}\\
\end{array}\right),$$
where  $A_{i, j}$  is a circulant matrix determined by the polynomial  $a_{i, j}(x)$, where  $1 \leq i \leq 2$  and  $1 \leq j \leq 2$. 

Just like the classical situation, a binary QECC $\mathscr{Q}$ of length $n$ is described as a $K$-dimensional subspace of $2^n$-dimensional Hilbert space $(\mathbb{C}^{2})^{\otimes n}$, where $\mathbb{C}$ represents complex field and $(\mathbb{C}^{2})^{\otimes n}$ is the $n$-fold tensor power of $\mathbb{C}^2$. If $K = 2^k$, $\mathscr{Q}$ is denoted as $[[n, k, d]]$, where $d$ is its minimum distance. 
 
The symplectic construction is one of the powerful methods for constructing QECCs. It establishes a correlation between symplectic self-orthogonal or dual-containing codes and QECCs, i.e., if there exists a symplectic self-orthogonal or dual-containing code, there exists a QECC with specific parameters.

\begin{lemma}(\cite{calderbank1998quantum}, Theorem 13)\label{symplectic_construction} Let  $\mathscr{C} \subset F_{2}^{2 n}$  be a symplectic dual-containing  $[2n, k]$  linear code, then there exists a pure QECC with parameters  $[[n, k-n, d]]$, where  $d$  is minimum symplectic weight of  $\mathscr{C} \backslash(\mathscr{C}^{\perp_{s}})$.
\end{lemma}

Typically, QECCs can also be derived from existing ones by the following propagation rules, which will be utilized later.

\begin{lemma}\label{DerivativeCodes}\label{morecodes}
(\cite{calderbank1998quantum}, Theorem 6) Suppose that an  $[[n, k, d]]$ pure QECC exists. Then the following QECCs exist.\\
(1)  $[[n, k-1, d]]$ for $k\ge 1$;\\
(2)  $[[n+1, k, d]]$ for $k>0$.
\end{lemma}

Throughout this paper, the QECCs we consider are all binary, all derived from binary symplectic dual-containing codes by employing the symplectic construction.
\section{New method of constructing QECCs }\label{sec3}
In this section, we propose a new method for constructing QECCs using 2-QC codes. 
First, we propose a suitable class of 2-QC codes and determine their parameters and symplectic dual algebraic structure. 
Further, we obtain the sufficient conditions for them to be symplectic dual-containing.
Then, by Lemma \ref{symplectic_construction}, a computer-supported method to produce quantum codes is derived.

Let  $g(x)=g_{0}+g_{1} x+\cdots+   g_{n-1} x^{n-1} \in \mathcal{R}$ and $[g(x)]$  represents vectors in  $F_{2}^{n}$  determined by the coefficient of  $g(x)$  in an ascending order.
Define $\bar{g}(x)=g_{0}+g_{n-1} x+g_{n-2} x^{2}+ \cdots+g_{1} x^{n-1}$.
Moreover, $g^{\perp}(x)=x^{\operatorname{deg}(h(x))} h\left(\frac{1}{x}\right)$  if  $g(x) h(x)=x^{n}-1$. 

\begin{proposition}\label{def1} Let  $\mathscr{C}$  be a  2-$\mathrm{QC}$  code over  $F_{2}$  of length $2n$ generated by  $(v(x)g_{1}(x), g_{1}(x))$  and  $(g_{2}(x), v(x) g_{2}(x))$, where $g_{1}(x),g_{2}(x),v(x)\in \mathcal{R}$ such that $g_{1}(x) \mid (x^{n}-1)$, $g_{2}(x) \mid(x^{n}-1)$, $\operatorname{gcd}(g_1(x), g_2(x))=1$ and  $\operatorname{gcd}(v(x)-1, x^{n}-1)=1$. Therefore, generator matrix of $\mathscr{C}$ is

$$G=\left(\begin{array}{cc}
G_{v1} & G_{1} \\
G_{2} & G_{v2} 
\end{array}\right),$$
where  $G_{1}$  and  $G_{2}$  are generator matrices of cyclic codes  $\langle g_1(x)\rangle$  and  $\langle g_2(x)\rangle$, respectively. $G_{v1}$ and $G_{v2}$ are $(n-deg(g_1(x))) \times n$ and $(n-deg(g_2(x))) \times n$ circulant matrices determined by  $v(x)g_1(x)$ and $v(x)g_2(x)$, separately.
\end{proposition}

\begin{proposition}\label{lamma_dimension}
The 2-$\mathrm{QC}$  code  $\mathscr{C}$  defined in Proposition \ref{def1} has parameters  $\left[2 n, 2 n-\operatorname{deg}\left(g_{1}(x)\right)-\operatorname{deg}\left(g_{2}(x)\right)\right].$
\end{proposition}

\begin{proof}
	Firstly, $\mathscr{C}$ is a quasi-cyclic code with a length of $2n$, whose generator matrix can be divided into two blocks, i.e. $\left(G_{v1}\quad G_{1}\right)$ and $\left(G_{2}\quad G_{v2} \right)$, above and below.
	Let $\mathscr{C}_{g_1}$ and $\mathscr{C}_{g_2}$ be one-generator quasi-cyclic codes separately generated by $\left(G_{v1}\quad G_{1}\right)$ and $\left(G_{2}\quad G_{v2} \right)$. Then it is easy to determine that $\mathscr{C}= \mathscr{C}_{g_1}\cup \mathscr{C}_{g_2}$ and dimensions of $\mathscr{C}_{g_1}$ and $\mathscr{C}_{g_2}$ are $n-\operatorname{deg}(g_{1}(x))$ and $n-\operatorname{deg}(g_{2}(x))$, respectively.
	Hence, $dim(\mathscr{C})\le dim(\mathscr{C}_{g_1})+dim(\mathscr{C}_{g_2})$.
	
In addition, if we assume that $dim(\mathscr{C})< dim(\mathscr{C}_{g_1})+dim(\mathscr{C}_{g_2})$, i.e., $\mathscr{C}_{g_1}\cap \mathscr{C}_{g_2}\ne \{\bm{0}\}$, $\bm{0}$ represents all zero codewords.
Then, there exist arbitrary polynomials $a(x)$ and $b(x)$ whose degree do not exceed $n-deg(g_1(x))-1$ and $n-deg(g_2(x))-1$, respectively, such that  $[a(x)(v(x)g_{1}(x), g_{1}(x))]+[b(x)(g_{2}(x), v(x) g_{2}(x))]=[0]$, where $[v(x)g_{1}(x), g_{1}(x)]\ne [0]$, $[g_{2}(x), v(x) g_{2}(x)]\ne[0]$.
Then $(x^{n}-1)\mid (a(x)v(x)g_{1}(x)+b(x)g_{2}(x))$  and  $(x^{n}-1)\mid (a(x) g_{1}(x)+b(x) v(x) g_{2}(x))$,  thereby $(x^n-1 ) \mid (a(x)g_1(x)+b(x)g_2(x))$.
We will discuss this in the following two cases.

\textbf{ Case a:} If $x^n-1=g_1(x)g_2(x)$.
Since $g_1(x),g_2(x)$ are monic divisors of $x^{n}-1$ and $\operatorname{gcd}\left(g_1(x), g_2(x)\right)=1$, it is easy to judge that if $(x^n-1) \mid (a(x)g_1(x)+b(x)g_2(x))$, there will $(x^n-1)\mid a(x)g_1(x) $ or $(x^n-1)\mid b(x)g_2(x) $, so $\mathscr{C}_{g_1}\cap \mathscr{C}_{g_2}=\{\bm{0}\}$, which contradicts the assumption.

\textbf{ Case b:} If $x^n-1\ne g_1(x)g_2(x)$. Let $g_3(x)=(x^n-1)/(g_1(x)g_2(x))=\xi(x) +\eta(x)$, where  $\xi(x),\eta(x),g_3(x)\in \mathcal{R}$, and $\xi(x),\eta(x)\ne0$.
Similar to Case a, we can determine that $(x^n-1) \mid (a(x)g_1(x)+b(x)g_2(x))$ holds if and only if 
$a(x)=m(x)g_2(x)\xi(x)$ and $b(x)=m(x)g_1(x)\eta(x)$, where $m(x)\in \mathcal{R}$.
However, in $\mathcal{R}$, there is

$\begin{array}{l} 
	a(x)v(x)g_1(x)+b(x)g_2(x) \\ 
	=m(x)g_1(x)g_2(x)((v(x)-1)\xi(x)+\xi(x)+\eta(x)(x)) \\ 
	=m(x)g_1(x)g_2(x)(v(x)-1)\xi(x)+x^n-1 \\ 
	=m(x)g_1(x)g_2(x)(v(x)-1)\xi(x). \\ 
\end{array} 
$

Since $g_3(x)\nmid\xi(x)$ and $\operatorname{gcd}(v(x)-1, x^{n}-1)=1$, $x^n-1 \nmid a(x)v(x)g_1(x)+b(x)g_2(x)$, thus $\mathscr{C}_{g_1}\cap \mathscr{C}_{g_2}=\{\bm{0}\}$.

In summary, we can conclude that $\mathscr{C}_{g_1}\cap \mathscr{C}_{g_2}=\{\bm{0}\}$. Therefore parameters of $\mathscr{C}$ is $\left[2 n, 2 n-\operatorname{deg}\left(g_{1}(x)\right)-\operatorname{deg}\left(g_{2}(x)\right)\right]$.
\end{proof}
\begin{lemma}\label{exchange_law}(\cite{galindo2018quasi}, Proposition 2): If  $f(x)$, $g(x)$  and  $h(x)$  are polynomials in  $F_{2}[x]$  of degree less than $n$, then the following equality of Euclidean inner products of vectors in  $F_{2}^{n}$  holds:
$$\langle[f(x) g(x)],[h(x)]\rangle_{e}=\langle[g(x)],[\bar{f}(x) h(x)]\rangle_{e}.$$
\end{lemma}

\begin{proposition} The symplectic dual code $\mathscr{C}^{\perp_{s}}$ of $\mathscr{C}$ is generated by pairs  $(g_{1}^{\perp}(x), \bar{v}(x) g_{1}^{\perp}(x))$  and  $(\bar{v}(x) g_{2}^{\perp}(x), g_{2}^{\perp}(x))$.
\end{proposition}
\begin{proof}Let  $\mathscr{C}_{0}$  denote linear code generated by pairs  $(g_{1}^{\perp}(x), \bar{v}(x) g_{1}^{\perp}(x))$  and  $(\bar{v}(x)g_{2}^{\perp}(x), g_{2}^{\perp}(x))$. 
Set $c_1=([a(x) v(x) g_{1}(x)+b(x) g_{2}(x)],[a(x)g_{1}(x)+b(x)v(x)$ $ g_{2}(x)])$,
$c_2=([c(x)g_{1}^{\perp}(x)+d(x)\bar{v}(x)g_{2}^{\perp}(x)],[c(x)\bar{v}(x)g_{1}^{\perp}(x)+d(x)g_{2}^{\perp}(x)])$, where $a(x),b(x),c(x),d(x)\in \mathcal{R}$.
Obviously, any codewords in $\mathscr{C}$ and $\mathscr{C}_0$ can be represented by $c_1$, $c_2$, respectively.
Then $\langle c_1, c_2\rangle_{s}$ is equal to
$$ \begin{aligned}
\langle[a(x)v(x)g_{1}(x)+b(x) g_{2}(x)], [c(x) \bar{v}(x) g_{1}^{\perp}(x)+d(x) g_{2}^{\perp}(x)]\rangle_{e} \\
-\langle[a(x) g_{1}(x)+b(x) v(x) g_{2}(x)],[c(x) g_{1}^{\perp}(x)+d(x) \bar{v}(x)g_{2}^{\perp}(x)]\rangle_{e} \\
=\langle[a(x)v(x)g_{1}(x)],[c(x) \bar{v}(x) g_{1}^{\perp}(x)]\rangle_{e} 
+\langle[a(x)v(x)g_{1}(x)],[d(x) g_{2}^{\perp}(x)]\rangle_{e} \\
+\langle[b(x) g_{2}(x)],[c(x) \bar{v}(x) g_{1}^{\perp}(x)]\rangle_{e} 
+\langle[b(x) g_{2}(x)],[d(x) g_{2}^{\perp}(x)]\rangle_{e} \\
-\langle[a(x) g_{1}(x)],[c(x) g_{1}^{\perp}(x)]\rangle_{e} 
-\langle[a(x) g_{1}(x)],[d(x) \bar{v}(x)g_{2}^{\perp}(x)]\rangle_{e} \\
-\langle[b(x) v(x) g_{2}(x)],[c(x) g_{1}^{\perp}(x)]\rangle_{e} 
-\langle[b(x) v(x) g_{2}(x)],[d(x) \bar{v}(x)g_{2}^{\perp}(x)]\rangle_{e}.
\end{aligned}$$

Since  $\langle g^{\perp}(x)\rangle$  is Euclidean dual code of cyclic code  $\langle g(x)\rangle$, above equation as can be simplified as 
$$ \begin{aligned}
=\langle[a(x)v(x)g_{1}(x)],[d(x) g_{2}^{\perp}(x)]\rangle_{e}
+\langle[b(x) g_{2}(x)],[c(x) \bar{v}(x) g_{1}^{\perp}(x)]\rangle_{e}
\\-\langle[a(x) g_{1}(x)],[d(x) \bar{v}(x)g_{2}^{\perp}(x)]\rangle_{e}
-\langle[b(x) v(x) g_{2}(x)],[c(x) g_{1}^{\perp}(x)]\rangle_{e}.
\end{aligned}$$

By Lemma \ref{exchange_law}, we have that this equation is equal to zero.
Because $\operatorname{gcd}(v(x)-1, x^{n}-1)=1$, then  $\bar{v}(x)-1$  and  $x^{n}-1$  are also relatively prime. 
It is analogous to the proof of Lemma \ref{lamma_dimension}, we can yield that  dimension of $\mathscr{C}_{0}$ is 
 $\operatorname{deg}\left(g_{1}(x)\right)+\operatorname{deg}\left(g_{2}(x)\right)$, which is actually equal to dimension of  $\mathscr{C}^{\perp_{s}}$.
Therefore,  $\mathscr{C}_{0}$  is symplectic dual code of $\mathscr{C}$.
\end{proof}

\begin{proposition} If  $g_{2}(x) \mid g_{1}^{\perp}(x)$, $\operatorname{gcd}\left(g_1(x), g_2(x)\right)=1$  and  $\bar{v}(x)=v(x)$, then our 2-QC code is symplectic dual-containing.
\end{proposition}
\begin{proof}
If  $g_{2}(x) \mid g_{1}^{\perp}(x)$ and  $\bar{v}(x)=v(x)$, then we can deduce that the following equations holds
$$(g_{1}^{\perp}(x), \bar{v}(x) g_{1}^{\perp}(x))=\frac{g_{1}^{\perp}(x)}{g_2(x)} \left(g_{2}(x), v(x) g_{2}(x)\right),$$ 
$$(\bar{v}(x)g_{2}^{\perp}(x),g_{2}^{\perp}(x))=\frac{g_{2}^{\perp}(x)}{g_1(x)}(v(x)g_{1}(x), g_{1}(x)).$$ 

Therefore, $\mathscr{C}^{\perp_{s}}$ can be linearly represented by $\mathscr{C}$, i.e. $\mathscr{C}^{\perp_{s}}\subset \mathscr{C}$, which yields the conclusion.
\end{proof}

\begin{theorem}\label{sconstruction} 
Let  $\mathscr{C}$  be a 2-QC code that proposed in Proposition  \ref{def1}.  If  $g_{2}(x) \mid g_{1}^{\perp}(x)$,  $\operatorname{gcd}(g_1(x),$ $ g_2(x))=1$  and  $\bar{v}(x)=v(x)$, then there exist a binary pure $[[n, n-\operatorname{deg}(g_{1}(x))-\operatorname{deg}(g_{2}(x)), d]]$ QECC, where  $d=\min \{w_{s}(\vec{u}) \mid \vec{u} \in \mathscr{C}\}$.
\end{theorem}
\begin{remark}
Let  $T_{1}$ and  $T_{2}$  be defining sets of cyclic  $\operatorname{codes}\left\langle g_{1}(x)\right\rangle$  and  $\left\langle g_{2}(x)\right\rangle$, respectively. If $T_{1} \cap\left(-T_{2}\right)=T_{1} \cap T_{2}=\emptyset$, then  $g_{2}(x) \mid g_{1}^{\perp}(x)$  and $\operatorname{gcd}\left(g_1(x), g_2(x)\right)=1$.
\end{remark}
\section{New binary QECCs}\label{sec4}
Since minimum distance determines error detection and correction capability, one of the most fundamental problems in coding theory is constructing QECCs with more considerable minimum distance for a given length $n$ and dimension $k$. This problem is challenging both in theory and practice.
For binary QECCs, Grassl collected binary QECCs with the best parameters so far and compiled them into an online code table \cite{Grassltable}, which is considered to be the most comprehensive table of binary QECCs. 

In this section, by virtue of  our construction method, 8 new binary QECCs are obtained, all of them improve the lower bounds on the minimum distance in \cite{Grassltable}, which are illustrated with examples below.
To save space, we represent the coefficient polynomials in ascending order, with indexes of elements to represent consecutive elements of same number. For example, polynomial $1 + x^2 + x^3 + x^4$ over $F_2$ is denoted as $101^3$. The parameters of corresponding codes are
computed by algebra software Magma \cite{magma1997}.
\begin{example}
Let $n=45$. Consider the 2-cyclotomic cosets modulo 45. 
Choosing  $T_{1}=C_{0} \cup C_{1}$ and  $T_{2}=C_{3}\cup C_{5}\cup C_{9}\cup C_{15}$  as defining sets of $\left\langle g_{1}(x)\right\rangle$  and  $\left\langle g_{2}(x)\right\rangle$, respectively. Then  $g_1(x)=1^{2}01^{2}0^{7}1^2$, $g_2(x)=1^{4}0^{2}1010^{4}1^{2}01$, and $v(x)=01^{2}01^{2}0^{34}1^{2}01^2$, which satisfy the conditions in Theorem \ref{sconstruction}. So a new binary QECC with parameters  $[[45,16,8]]$  can be provided, whose weight distribution is $w(z)=1+2970 z^{8}+32580 z^{9}+352620 z^{10}+\cdots+ 5502836658465 z^{45}$. 
Observe that a code with parameter  $[[45,16,7]]$ is the best-known binary QECC with length $45$ and dimension $16$ in \cite{Grassltable}, therefore the current record of corresponding minimum distance can be improved to $8$.
\end{example}
\begin{example}
Let $n=51$. Consider the 2-cyclotomic cosets modulo 51. 
Choosing  $T_{1}=C_{0}$ and  $T_{2}=C_{1} \cup C_{3}\cup C_{5}$  as defining sets of     $\left\langle g_{1}(x)\right\rangle$ and  $\left\langle g_{2}(x)\right\rangle$, respectively.
Then  $g_1(x)=1^2$, $g_2(x)=1^{2}0^{3}10^{2}1^{2}0^{2}1^{2}01010^{2}1^{3}01$, and $v(x)=0^{4}1^{2}010101^{2}01^{2}0^{3}1^{2}01010^{2} 10101^{2}0^{3}1^{2}01^{2}0$ $10101^2$, which satisfy the conditions in Theorem \ref{sconstruction}. So a new binary QECC with parameters  $[[51,26,7]]$  can be produced, whose weight distribution can be written as  $w(z)=1+7854 z^{7}+126021z^{8}+1783827z^{9}+\cdots+ 64185081779884617z^{51}$.
In \cite{Grassltable}, the best-known binary QECC with length $51$ and dimension $26$ is $[[51,26,6]]$, therefore the current record of corresponding minimum distance can be improved to $7$.
\end{example}

\begin{remark} We also derive another 6 new binary QECCs with parameters $[[47,23,7]]$,   $[[55,25,8]]$,  $[[73,45,7]]$, $[[91,60,7]]$, $[[93,63,7]]$, $[[105,74,7]]$, respectively, which all better than the corresponding QECCs appeared in Grassl's code table \cite{Grassltable}, i.e., those QECCs have improved the best-known minimum distance for same length and dimension. Details are as follows.

(1) $[[47,23,7]]$: $g_1(x)=10^{3}1^{2}0^{3}1^{3}01^{2}01^{3}01^4$, $g_2(x)= 1^2$, and $v(x)=0^{2}1^{5}01^{3}0^{3}1^{2}0^{16}1^{2}0^{3}1^{3}01^5$.  Weight distribution  $w(z)=1+ 8131 z^{7}+124221 z^{8}+1609280 z^{9}+\cdots+   1584816838311358 z^{47}$.

(2) $[[55,25,8]]$: $g_1(x)=101^{2}01^{2}010^{2}101^{3}0^{2}1^3$, $g_2(x)=11$, and $v(x)=010101^{4}010^{2}1$ $0^{28}10^{2} 101^{4}0101$.  Weight distribution  $w(z)= 1+7810 z^{8}+120285 z^{9}+1625580 z^{10}+\cdots+   162468493708434387 z^{55}$.

(3)	 $[[73,45,7]]$: $g_1(x)=1^{3}01^{3}0^{2}10^{2}1^{3}01^3$, $g_2(x)=1^{2}0 10101^4$, and $v(x)=01010^{2}1^{2}0^{58} 1^{2}0^{2}101$.  Weight distribution  $w(z)=1+13067 z^{7}+330252 z^{8}+7138451 z^{9}+\cdots+   25177448479 3909182856047064 z^{73}$.	 	

(4)	 $[[91,60,7]]$: $g_1(x)=10^{2}1^{4}010^{3}1$, 
$g_2(x)=1^{2}0^{5}1^{4}0 10^{3}1^{2}01$, and $v(x)=01^{2}0^{2}10^{80}10^{2}1^2$.  Weight distribution  $w(z)=1+7930 z^{7}+252070 z^{8}+7192367 z^{9}+\cdots+   12192824252072319962279913005679504 z^{91}$.	 

(5) $[[93,63,7]]$: 
$g_1(x)=1^{3}01^{2}0^{4}1^{3}0101010101^{2}01$, $g_2(x)=1^{3}01^2$,
and $v(x)=0^{3}$ $1010^{2}10^{2}1010101^{4}01^{2}01010^{8}101^{2}0^{2}1^{10}0^{2}1^{2}010^{8}10101^{2}01^{4}0 101010^{2}10^{2}101$. Weight distribution  $w(z)=  1+20274 z^{7}+61975 z^{8}+17662002 z^{9}+\cdots+ 219470836537$ $301759730706283742251295 z^{93}$.	 

(6) $[[105,74,7]]$: $g_1(x)=1^{2}01^{3}0^{3}10^{2}1$, 
$g_2(x)=1^{2}01 0^{4}1^{3}01^8$, and $v(x)=0^{2}1010^{4}10^{86}10^{4}101$.  Weight distribution  $w(z)=1+21855 z^{7}+854385 z^{8}+27543740 z^{9}+\cdots+   58317900420110 092238457799326599725172775 z^{105}$.	 	

\end{remark}

According to Lemma \ref{morecodes}, we can deduce another 36 new binary
QECCs from QECCs above. As shown in Table \ref{tab:my-table}, their parameters also improve the lower bounds on the minimum distance in Grassl's table \cite{Grassltable}.
\begin{table}[ht]
	\caption{New binary QECCs}
	\label{tab:my-table}
	\begin{center}
		{\small
			\setlength\tabcolsep{1.5pt}
			\begin{tabular}{cccccc}
				\hline
				NO. &    Our QECCs    & QECCs in \cite{Grassltable} & NO. &    Our QECCs     & QECCs in  \cite{Grassltable} \\ \hline
				 1  & $[[46,16,8]]$ &       $[[46,16,7]]$       & 19  & $[[76,45,7]]$  &       $[[76,45,6]]$        \\
				 2  & $[[47,16,8]]$ &       $[[47,16,7]]$       & 20  & $[[76,44,7]]$  &       $[[76,44,6]]$        \\
				 3  & $[[48,23,7]]$ &       $[[48,23,6]]$       & 21  & $[[76,43,7]]$  &       $[[76,43,6]]$        \\
				 4  & $[[49,23,7]]$ &       $[[49,23,6]]$       & 22  & $[[77,45,7]]$  &       $[[77,45,6]]$        \\
				 5  & $[[50,23,7]]$ &       $[[50,23,6]]$       & 23  & $[[77,44,7]]$  &       $[[77,44,6]]$        \\
				 6  & $[[51,23,7]]$ &       $[[51,23,6]]$       & 24  & $[[78,45,7]]$  &       $[[78,45,6]]$        \\
				 7  & $[[51,25,7]]$ &       $[[51,25,6]]$       & 25  & $[[91,59,7]]$  &       $[[91,59,6]]$        \\
				 8  & $[[51,24,7]]$ &       $[[51,24,6]]$       & 26  & $[[91,58,7]]$  &       $[[91,58,6]]$        \\
				 9  & $[[56,25,8]]$ &       $[[56,25,7]]$       & 27  & $[[92,60,7]]$  &       $[[92,60,6]]$        \\
				10  & $[[57,25,8]]$ &       $[[57,25,7]]$       & 28  & $[[92,59,7]]$  &       $[[92,59,6]]$        \\
				11  & $[[73,44,7]]$ &       $[[73,44,6]]$       & 29  & $[[93,62,7]]$  &       $[[93,62,6]]$        \\
				12  & $[[73,43,7]]$ &       $[[73,43,6]]$       & 30  & $[[94,63,7]]$  &       $[[94,63,6]]$        \\
				13  & $[[74,45,7]]$ &       $[[74,45,6]]$       & 31  & $[[95,63,7]]$  &       $[[95,63,6]]$        \\
				14  & $[[74,44,7]]$ &       $[[74,44,6]]$       & 32  & $[[105,73,7]]$ &       $[[105,73,6]]$       \\
				15  & $[[74,43,7]]$ &       $[[74,43,6]]$       & 33  & $[[106,74,7]]$ &       $[[106,74,6]]$       \\
				16  & $[[75,45,7]]$ &       $[[75,45,6]]$       & 34  & $[[106,73,7]]$ &       $[[106,73,6]]$       \\
				17  & $[[75,44,7]]$ &       $[[75,44,6]]$       & 35  & $[[107,74,7]]$ &       $[[107,74,6]]$       \\
				18  & $[[75,43,7]]$ &       $[[75,43,6]]$       & 36  & $[[107,73,7]]$ &       $[[107,73,6]]$       \\ \hline
			\end{tabular}}\end{center}
\end{table}
\section{Conclusion}\label{sec5}
This paper investigates a suitable class of 2-QC codes, determining their parameters and symplectic dual algebraic structure, as well as the symplectic dual-containing conditions. Subsequently, we propose a new method for constructing QECCs from these 2-QC codes. 
Finally, a total of 44 record-breaking binary QECCs are derived by the symplectic construction or propagation rules. The results of this paper show that it is promising to study QC codes for producing new QECCs.
\section*{Acknowledgments}
This work is supported by National Natural Science Foundation of China
under Grant No. U21A20428, 11901579, 11801564, Natural Science Foundation of
Shaanxi under Grant No. 2021JM-216, 2021JQ-335, 2022JQ-046 and the Graduate
Scientific Research Foundation of Fundamentals Department of Air
Force Engineering University.

\bibliography{sn-bibliography}


\end{document}